\documentclass[runningheads,a4paper]{llncs}

\usepackage{amssymb}
\usepackage{graphicx}
\usepackage[fleqn]{amsmath}

\usepackage{url}
\urldef{\mailsa}\path|Rupei.Xu@utdallas.edu|
\urldef{\mailsb}\path|farago@utdallas.edu|
\newcommand{\keywords}[1]{\par\addvspace\baselineskip
\noindent\keywordname\enspace\ignorespaces#1}

\begin{document}

\mainmatter  

\title{The Landscape of Minimum Label Cut (Hedge Connectivity) Problem}

\titlerunning{The Landscape of Minimum Label Cut (Hedge Connectivity) Problem}

%
%
\author{Rupei Xu \and Andr\'{a}s Farag\'{o}}
\authorrunning{Rupei Xu \and Andr\'{a}s Farag\'{o}}

\institute{The University of Texas at Dallas\\
\mailsa\\
\mailsb\\
}

%
%

\toctitle{The Landscape of Minimum Label Cut (Hedge Connectivity) Problem}
\tocauthor{R.X and A. F}
\maketitle

\begin{abstract}
	
Minimum Label Cut (or Hedge Connectivity) problem is defined as follows: given an undirected graph $G=(V, E)$ with $n$ vertices and $m$ edges, in which, each edge is labeled (with one or multiple labels) from a label set $L=\{\ell_1,\ell_2, ..., \ell_{|L|}\}$, the edges may be weighted with weight set $W =\{w_1, w_2, ..., w_m\}$, the label cut problem(hedge connectivity) problem asks for the minimum number of edge sets(each edge set (or hedge) is the edges with the same label) whose removal disconnects the source-sink pair of vertices or the whole graph with minimum total weights(minimum cardinality for unweighted version). This problem is more general than edge connectivity and hypergraph edge connectivity problem and has a lot of applications in MPLS, IP networks, synchronous optical networks, image segmentation, and other areas. However, due to limited communications between different communities, this problem was studied in different names, with some important existing literature citations missing, or sometimes the results are misleading with some errors. In this paper, we make a further investigation of this problem, give uniform definitions, fix existing errors, provide new insights and show some new results. Specifically, we show the relationship between non-overlapping version(each edge only has one label) and overlapping version(each edge has multiple labels), by fixing the error in the existing literature; hardness and approximation performance between weighted version and unweighted version and some useful properties for further research. 

\end{abstract}

\section{Introduction}

In the application of telecommunication networks, network security, image segmentation, some edges are correlated with each other and share the risk of a common failure.  Those edges can be associated with labels (or colors) which
partition the set of edges into categories. The label cut problem(hedge connectivity) problem asks for the minimum number of edge sets(each edge set (or hedge) is the edges with the same label) whose removal disconnect the source-sink pair of vertices or the whole graph with minimum total weights(minimum cardinality for unweighted version). This problem is more general than the traditional min cut or hypergraph min cut problem, since the later ones are just special cases of the label cut problem. 

\subsection{Brief Literature Review}

\subsubsection{Minimum Label $s-t$ Cut problem}

Coudert, Datta, et al. \cite{cd} first considered the Minimum Label $s-t$ Cut problem, in which it was called Minimum $s-t$ Color Cut problem. Jha, Sheyner and Wing \cite{js} observed that the Minimum Label $s-t$ Cut problem is NP-hard since the Minimum Hitting Set problem can be reduced to it. Since there is a simple duality between Minimum Hitting Set and Minimum Set Cover, the approximation algorithm for Set Cover can be transformed to an approximation algorithm for Label Cut with approximation guarantee of $1+\ln|L|$. Zhang, Cai, et al. \cite{zc} gave the first non-trivial approximation algorithm for the Minimum Label $s-t$ Cut problem on general graphs, where a polynomial time $O(m^{1/2})$-approximation algorithm was given. For the approximation hardness, they showed that the Minimum Label $s-t$ Cut problem can not be approximated within $2^{\log^{1-\frac{1}{\log\log^cn}}n}$ for any constant $c<1/2$ unless $P=NP$. Tang and Zhang \cite{tz} further improved the approximation result, using linear programming techniques, they got $\min\{O((m/OPT)^{1/2}), O(n^{2/3}/OPT^{1/3})\}$-approximation algorithm.  Zhang \cite{zhang14} gave a combinatorial $\ell_{max}$-approximation algorithm for the Label $s-t$ Cut problem, where $\ell_{max}$ is the maximum $s-t$ length. 

Broersma et al. \cite{br} devised exact algorithm for the Minimum Label $s-t$ Cut problem with running time $O(n^2|L|!)$, where $L$ denotes the set of labels. 

Fellows, Guo and Kanj \cite{fgk} did the parameterized complexity of the Minimum Label $s-t$ Cut problem, they showed that it is $W[2]$-hard on graphs with pathwidth at most $3$ parameterized by the number of used labels $d$, $W[1]$-hard on graphs with pathwidth at most $4$, parameterized by the solution size.

\subsubsection{Minimum Label $s-t$ Cut problem with Label Overlaps}
Farag\'{o} \cite{fa} first studied the Path Vulnerability problem with Label Overlaps. He argued that for a path the important thing is not that how many failure sets cover the entire path. It is more important that how many failure sets intersects with it, i.e., contains a link from the path. To characterize the expressibility of the vulnerability measure(no matter what kind of path metric), he gave a Path Metric Representation Theorem to show the path vulnerability. However, the complexity of the Minimum Label $s-t$ Cut problem with Label Overlaps is unknown. The Path Vulnerability problem with Label Overlaps is a dual of the Minimum Label $s-t$ Cut problem with Label Overlaps. 

\subsubsection{Minimum Label Cut problem}

Minimum Label Cut problem is polynomial-time solvable in several special cases, including graphs with bounded treewidth, planar graphs, and instances with bounded label frequency \cite{zhang14}. Ghaffari, Karger and Panigrahi \cite{gkp} studied the Minimum Label Cut problem (they call this problem Hedge Connectivity), in which they gave a polynomial-time approximation scheme and a quasi-polynomial exact algorithm. They mentioned the Minimum Label Cut problem with Label Overlaps can be reduced to the Minimum Label Cut problem by replacing the overlap edge with a rainbow path, that operation won't change the hedge connectivity. However, their argument does not hold(one can find a more detailed argument in section 3), thus their algorithms only work for the non-overlapping case. 

\subsubsection{Minimum Label Cut problem with Label Overlaps}

Farag\'{o} \cite{fa} first studied the Minimum Label Cut problem with Label Overlaps. He showed that for an input positive integer $p$ it is NP-complete to decide whether the label cut(with label overlaps) less than or equal to $p$ exists. 

\subsection{Main Results and Organization}

This paper gives uniform notations for four different versions of the minimum label cut(hedge connectivity) problem of whether it is for a source-sink label cut or global cut, and whether edge labels have overlaps or not(all the four versions can be either unweighted or weighted depends on whether edge's weights are unique or not). First, we show the overlapping version can be transformed into a weighted non-overlapping version by a polynomial time operation. Then, we show the hardness and approximation performance of unweighted version and weighted version. Moreover, some useful properties of minimum label cut problem and further open problems are discussed.

\section{Preliminaries}

In this paper, unless specific mention, all original graphs are simple, i.e. they have no self-loop nor multiple edges and connected. If you see one edge has multiple colors, that means that a single edge has multiple labels (label overlaps). 

Given an undirected graph $G=(V, E)$, a label set $L=\{\ell_1,\ell_2,...,\ell_{|L|}\}$, an edge weight set $W =\{w_1, w_2, ..., w_m\}$, the four versions of minimum label cut problem(weighted or unweighted based on the edge weight set) are defined as follows: 

\begin{definition} \textbf{(Minimum Label $s-t$ Cut problem)}
	
Each edge has a label from the label set, and there are a source vertex $s\in V$ and a sink vertex $t\in V.$ A label $s-t$ cut is a subset of labels such that the removal of all edges with these labels from $G$ disconnects $s$ and $t$ in $G.$ The Minimum Label $s-t$ Cut problem is to find a label $s-t$ cut of minimum size. 
\end{definition}

\begin{definition} \textbf{(Minimum Label $s-t$ Cut problem with Label Overlaps)}
	
Each edge has one or multiple labels from the label set, and there are a source vertex $s\in V$ and a sink vertex $t\in V.$ A label $s-t$ cut is a subset of labels such that the removal of all edges with these labels from $G$ disconnects $s$ and $t$ in $G.$ The Minimum Label $s-t$ Cut problem is to find a label $s-t$ cut of minimum size. 
\end{definition}

\begin{definition} \textbf{(Minimum Label Cut problem)}
	
Each edge $e\in E$ has a label from the label set. A label cut is a subset of labels such that the removal of all edges with these labels from $G$ partitions $G$ into at least two connected components. The Minimum Label Cut problem is to find a label cut of minimum size. 
\end{definition}

\begin{definition} \textbf{(Minimum Label Cut problem with Label Overlaps)}
	
Each edge $e\in E$ has one or multiple labels from the label set. A label cut is a subset of labels such that the removal of all edges with these labels from $G$ partitions $G$ into at least two connected components. The Minimum Label Cut problem with Label Overlaps is to find a label cut of minimum size. 
\end{definition}

\section{Transform Overlapping Version into Non-overlapping Version}

In paper \cite{gkp}, it says "We also note that by insisting on the fact that the hedges are disjoint, we are not losing any generality. If hedges overlap on an edge, modeling the fact that the edge can fail as part of multiple groups, we can replace the edge by a path where each edge on the path belongs to a unique hedge. This transformation does not affect the hedge connectivity of the graph since removing any of the hedges containing the original edge disconnects the path in the transformed graph."

To illustrate the above operation, we show the following example. 

\begin{figure} \label{fig1}
	\centering
	\includegraphics[width=0.9\linewidth]{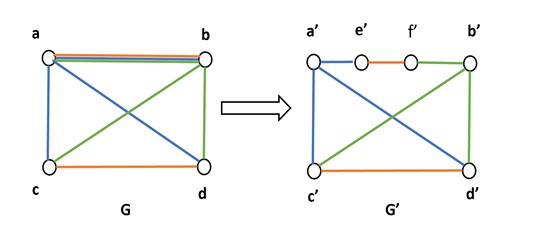}
	\caption{Replace Edge with Label Overlaps with Rainbow Path}
	\label{fig:operation}
\end{figure}

In Figure 1. the edge $E_{ab}$ has three colors, according to the paper \cite{gkp}, it can be replaced by a rainbow path from $a$ to $b$, the path has three colors. However, this operation is not hedge connectivity preserving. The error is that when we replace an edge with a path, such that each overlapping hedge that contained the original edge now only contains one edge of the path, then any new vertex of the path can be separated from the rest of the graph by removing the two hedges that contain its adjacent edges. Therefore, the hedge connectivity after the transformation drops to at most two, no matter how much it was originally. Thus, this transformation does not preserve hedge connectivity.

Let $D_L(v)$  be the label degree of vertex $v$, which is the number of different labels of adjacent edges of vertex $v$. 

\begin{property}
	The Minimum Label Cut of a graph is no larger than $\min\{D_L(v_i)\}$, where $i=1,2,...,n.$
\end{property}

In Figure 1. the original graph, $D_L(a)=D_L(b)=D_L(c)=D_L(d)=3$, but in the new graph, $D_L(a')=D_L(b')=1$, $D_L(e')=D_L(f')=2$, $D_L(c')=D_L(d')=3$, thus $\min_{v\in G}D_L(v)=3$ and $\min_{v\in G'}D_L(v)=1$. One can easily get the hedge connectivity of $G'$ is $1$, but the hedge connectivity of $G$ is definitely not $1$, thus that operation is not hedge connectivity preserving. 

Then, what operation can preserve hedge connectivity? We show it in the following: 

\begin{lemma}
	If there is one edge $e$ labeled by two different labels $L_i$ and $L_j$, then these two labels are correlated with each other, i.e. if there is one edge labeled by $L_i$ is removed(which leads all edges labeled by $L_i$ are removed), all other edges labeled $L_j$ are also removed consequently.  
\end{lemma}

\begin{corollary}
	If there is one edge $e$ labeled by two different labels $L_i$ and $L_j$ (it may be also labeled by other labels other than $L_i$ and $L_j$ ), then all other edges labeled by $L_i$ or $L_j$ can be relabeled by a new created label $L_t'$ with weight $1+1=2$, this operation preserves the hedge connectivity. 
\end{corollary}

In a graph with multiple edges which have label overlaps, we use disjoint-set data structure(also called a union-find or merge–find set data structure) to count the minimum relabel operation. Let's name the following procedure as operation $K$.

\begin{itemize}
	\item \textbf{Step 1.} \\ List all edges which have label overlaps, each of them is represented as a set with edge labels as its elements. 
	\item \textbf{Step 2.} \\ If two sets have common elements, merge them as one union set. 
	\item \textbf{Step 3.} \\ Relabel each union set elements with one new created label. 
\end{itemize}

\begin{figure}\label{fig2}
	\centering
	\includegraphics[width=0.8\linewidth]{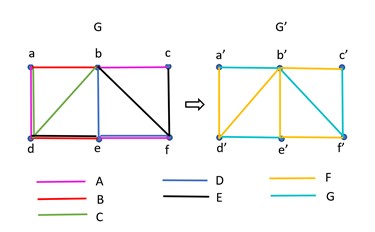}
	\caption{Operation $K$}
	\label{fig:operationnew}
\end{figure}

For example, in Figure 2., Graph $G$ has five kinds of labels and three edges have label overlaps. The three edges can be represented as three sets, their labels are its elements: 
$\{A, C\}$, $\{B, E\}$, $\{A,D\}$. After merging sets with common elements into union sets, we can get new sets: $\{A, C, D\}$, $\{B, E\}$. Label edges with labels in set $\{A, C, D\}$ as $F$(yellow) with weight $3$ and label edges with labels in set $\{B, E\}$ as $G$ (blue) with weight $2$. Eventually, one can get $G'$, which preserves the hedge connectivity of $G$. 

\begin{theorem}
	After taking operation $K$ for the graph with label overlaps, (1) each edge only has one label, (2) the weight of each hedge is no larger than $\max_{i=1}^m\{D_L(v_i)\}$ of the original graph, the total number of hedges in the new graph is no larger than the total number of hedges in the original graph, (3) the total weight of hedges in the new graph is exactly equal to the number of hedges of the original graph. 
\end{theorem}

\begin{proof}
	(1) is obvious just following operation $K$.  
	
	If there is one hedge with weight larger than $\max_{i=1}^m\{D_L(v_i)\}$, then in the original graph, there must be a vertex $v_j$, in which  $D_L(v_j)\geq 1+\max_{i=1}^m\{D_L(v_i)\}$ (since the original graph is connected and this hedge in the original graph is the collection of a number of labels as same as the weight in the new graph), which contradicts the definition of  $\max_{i=1}^m\{D_L(v_i)\}$. An obvious consequence is that the total number of hedges in the new graph is no larger than the total number of hedges in the original graph. Therefore (2) is true. 
	
	As after label replacement step, correlated labels(have nonempty overlaps in some edges) is replaced by a new label with weight to be the number of such correlated labels, each label with other label correlations only replaced by one new label, thus the total weight of hedges in the new graph is exactly equal to the total number of hedges in the original graph. Thus (3) is proved. 
\end{proof}

\section{Between Weighted Version and Unweighted Version}

\subsection{Computational Complexity with Label Overlaps }

It is known that the Minimum Label Cut problem with Label Overlaps is NP-hard. 

\begin{theorem} \cite{fa}
It is NP-hard to find the Minimum Label Cut problem with Label Overlaps.
\end{theorem}

The basic idea of the proof in \cite{fa} is via a reduction from the Set Cover problem to the decision version of the problem whether the Label Cut with label overlaps is less than or equal to a positive integer $p$ exists. 

 However the hardness of Minimum Label $s-t$ Cut problem with Label Overlaps is still open, we give a proof that it is also NP-hard.  

\begin{theorem}
It is NP-hard to find the Minimum Label $s-t$ Cut problem with Label Overlaps. 
\end{theorem}

\begin{proof}
Given a set $U=\{u_1, u_2, u_3,..., u_n\}$ and $S_1, S_2, ..., S_m$ are subsets of $U$ and integer $l.$ The hitting set problem is to find a subset of $U$ with size at most $l$ which has intersection with all $S_1, S_2,..., S_m.$ It is known that the hitting set problem is NP-complete, which is among the original 21 NP-complete problem list of Karp \cite{karp}. According to Menger's theorem, the size of the minimum edge cut $C_{st}$ for $s$ and $t$ is equal to the maximum number of pairwise internally edge-disjoint paths from $s$ to $t$. The minimum edge cut $C_{st}$ for $s$ and $t$ can be found in polynomial time. Then one can apply the operation $K$ introduced earlier to relabel the graph. Let's make a construction to make $U=L'$, $m=|C_{st}|$, each $S_i$ is corresponding to one $st$-path of the graph. The hitting set intersecting all $S_i$ is exactly the number of hedge edges removal to disconnect all the $|C_{st}|$ paths of $s$ and $t$. 
\end{proof}

\subsection{Approximation Results with Label Overlaps}

\begin{theorem}	
Minimum Label $s-t$ Cut problem with Label Overlaps is APX-hard. 
\end{theorem}

\begin{proof}
After doing the operation $K$ as the preprocessing procedure, the Minimum Label $s-t$ Cut problem with Label Overlaps transformed into the Weighted Minimum Label $s-t$ Cut problem. This operation preserves the hedge connectivity, and can be done in polynomial time. As the Weighted Minimum Label $s-t$ Cut problem is APX-hard \cite{zc}, thus the Minimum Label $s-t$ Cut problem with Label Overlaps is also APX-hard.
\end{proof}

Based on the existing results, in addition to some lemmas, one can easily get corollaries for label overlapping problems. 

\begin{corollary}
Minimum Label $s-t$ Cut problem with Label Overlaps does not admit PTAS even in graphs
$G$ with $l_{max} = 2$, $f_{max} = 4$, and $tw(G) = 2$, where $l_{max}$ is longest path between $s$ and $t$, $f_{max}$ is the label frequency, i.e. the number of appearances of a label in the input graph and $tw(G)$ is the treewidth of $G.$
\end{corollary}

\begin{corollary}
For any constant $c<1/2,$ the Minimum Label $s-t$ Cut problem with Label Overlaps can not be approximated within $2^{\log^{1-\frac{1}{\log\log^cn}}n}$ in polynomial time unless P=NP, where $n$ is the input length of the problem. 
\end{corollary}

Due to the breakthrough work on the resolution of $2$-to-$2$ game conjecture \cite{2to2}, one can get a better approximation lower bound than the previous best known result of Vertex Cover Problem \cite{vc}. 
\begin{lemma} \cite{2to2}
It is NP-hard to approximate Vertex Cover problem within a factor of $\sqrt{2}$.
\end{lemma}

\begin{corollary}
Label $s-t$ cut problem with the longest path between $s$ and $t$ is bounded by $2$ cannot be approximated within $(\sqrt{2}-\epsilon)$ for any constant $\epsilon>0$. 
\end{corollary}

\begin{corollary}
If the longest path between $s$ and $t$ is bounded by $2$, it is NP-hard to approximate Minimum $s-t$ Label Cut problem with Label Overlaps within a factor of $\sqrt{2}$. 
\end{corollary}

\begin{corollary}
Minimum Label Cut problem with Label Overlaps can not be approximated within $(1-o(1))\ln n$ unless $P=NP$.
\end{corollary}

\section{Submodularity}

\begin{figure} \label{fig3}
	\centering
	\includegraphics[width=0.4\linewidth]{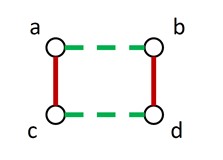}
	\caption{Hedge Cut Example}
	\label{fig:hedge}
\end{figure}

Figure 3. is the example in paper \cite{gkp} to argue the hedge cut function is not submodular. In that example, $f(\{a,b\}-f(\{a\})=1-2< 2-1=f(\{a,b,c\})-f(\{a,c\}).$ 

In telecommunication networks, Shared Risk Resource Group(SRG) is exactly the problems studied above. An SRG failure makes multiple circuits go down because of the failure of a common resource those networks share. There are different shared risk groups: Shared risk link group (SRLG) and Shared risk node group (SRNG). The paper \cite{gkp} was studying the SRLG problem, however, the example it showed was actually SRNG. 

A more proper way to define the hedge cut function is by the number of labels of the cut edge set. Given a subset $E'\subseteq E$ of edges, let $g(E')$ to be the number of labels appeared in $E'$. The minimum hedge cut $E'$: 
$$\min\{g(E')|E'\subseteq E, G(V, E\setminus E')~disconnected\}.$$

\begin{property}
	Let $E'\subseteq E$ and $E''\subseteq E$, the function of the number of labels in Minimum Label Cut Problem is submodular:
	
	$$g(E')+g(E'')\geq g(E'\cup E'')+g(E'\cap E'').$$
\end{property}

Above properties also hold in minimum label cut problem with label overlaps. 

\begin{property}
The Minimum Label Cut with Label Overlaps of a graph is no larger than $\min\{D_L(v_i)\}$, where $i=1,2,...,n.$
\end{property}

\begin{property}
Let $E'\subseteq E$ and $E''\subseteq E$, the function of the number of labels in Minimum Label Cut Problem with Label Overlaps is submodular:

$$g(E')+g(E'')\geq g(E'\cup E'')+g(E'\cap E'').$$

\end{property}

\section{Further Discussion and Open Problems}

There are some interesting open problems for further research. One may try different approaches for the minimum label cut problem by streaming, graph sparsification method, or distributed algorithm. It is also interesting to study the problem with special hedge structures such as tree, path, cycle. The Menger's type theorem for this problem is still missing, it is open whether Menger's type theorem exists for the minimum label cut problem. Moreover, those results may have connections with other cut problems such as Label Multicut, Label Multiway Cut, and Label $k$-Cut. Last but not least, it is still open whether the (global) minimum label cut(nonoverlapping version) problem is polynomial-time solvable or not.

\paragraph{\textbf{Acknowledgements}}

Rupei Xu would like to thank Dr. Debmalya Panigrahi for kind instructions and helpful discussions on graph connectivity topic during her visit to Simons Institute for the Theory of Computing of UC Berkeley. Dr. Debmalya Panigrahi's talks about hedge connectivity in the  9th Workshop on Flexible Network Design and 2018 INFORMS Annual Meeting inspired Rupei Xu to start this research. Dr. Kyle J. Fox and Dr. L\'{a}szl\'{o} A. V\'{e}gh also helped to discuss this problem.

\end{document}